\documentclass[conference,letterpaper]{IEEEtran}

\IEEEoverridecommandlockouts

\usepackage[utf8]{inputenc} 
\usepackage[T1]{fontenc}
\usepackage{url}
\usepackage{ifthen}
\usepackage{cite}
 \usepackage[cmex10]{amsmath} 
 \usepackage{mathrsfs}
\usepackage{amssymb,bm, bbm}

 \usepackage{tikz}
 \usepackage{graphicx}
 \usepackage{pgfplots}
\usepackage{subcaption}

\usepackage[
  letterpaper,
  left=0.64in,
  right=0.64in,
  top=0.76in,
  bottom=1.1in
]{geometry}
\usepackage{epsfig}
\usepackage{times}
\usepackage{float}
\usepackage{afterpage}
\usepackage{amsmath}
\usepackage{amstext}
\usepackage{amssymb,bm, bbm}
\usepackage{latexsym}
\usepackage{color}
\usepackage{graphicx}
\usepackage{amsmath}
\usepackage{amsthm}
\usepackage{graphicx}
\usepackage{multicol}
\usepackage{lipsum}
\usepackage{dblfloatfix}
\usepackage{mathrsfs}
\usepackage{cite}
\usepackage{tikz}
\usepackage{pgfplots}
\pgfplotsset{compat=newest}
 \usepackage{enumitem}
\usepackage{makecell}
\usepackage{float}
\restylefloat{table}

\allowdisplaybreaks
\usepackage{url}

\newcommand{\D}{\mathsf{D}}
\newcommand{\supp}{\mathrm{supp}}
\newcommand{\KL}{\mathsf{D}}

\newcommand{\rme}{\mathrm{e}}

\newcommand{\bbP}{\mathbb{P}}

\newcommand{\bbR}{\mathbb{R}}

\newcommand{\cA}{\mathcal{A}}

\newcommand{\kl}[2]{{\D}\left(\left.#1 \, \right\| #2 \right)}

\theoremstyle{mystyle}
\newtheorem{theorem}{Theorem}
\theoremstyle{mystyle}
\newtheorem{lemma}{Lemma}
\theoremstyle{mystyle}
\newtheorem{prop}{Proposition}
\theoremstyle{mystyle}
\theoremstyle{mystyle}
\theoremstyle{remark}
\theoremstyle{mystyle}
\theoremstyle{mystyle}
\theoremstyle{mystyle}
\theoremstyle{discussion}
\theoremstyle{mystyle}
\theoremstyle{mystyle}

 \usepackage{enumitem}

\begin{document}

\title{An Improved Lower Bound on  Support Size of Capacity-Achieving Inputs for the Binomial Channel: Extended version}

\author{
\IEEEauthorblockN{Mohammadamin Baniasadi$^{\dagger}$, Antonino Favano$^{*}$, Luca Barletta$^{*}$,  and Alex Dytso$^{**}$}
$^{\dagger}$ University of California, Davis, USA. Email: mbaniasadi@ucdavis.edu \\
$^{*}$ Politecnico di Milano, Milan, Italy. Email: antonino.favano@polimi.it \\
$^{*}$ Politecnico di Milano, Milan, Italy. Email: luca.barletta@polimi.it \\
$^{**}$ Qualcomm Flarion Technologies, Bridgewater, USA.
Email: odytso2@gmail.com 
}

\maketitle

\begin{abstract}
We study the binomial channel and the structure of its capacity-achieving input and output distributions. It is known that the capacity-achieving input distribution is discrete and supported on finitely many points. The best previously known bounds show that the support size of the capacity-achieving distribution is lower-bounded by a term of order $\sqrt n$ and upper-bounded by a term of order $n/2$, where $n$ is the number of trials.

In this work, we derive a new lower bound on the support size of order $\sqrt{n\log\log n}$, up to explicit constants. The proof consists of three main steps. First, we derive new upper and lower bounds on the capacity with a gap that vanishes as $n\to\infty$, which yields
$C(n)=\frac12\log\frac{n\pi}{2e}+o(1)$.
Second, we show that the Beta-binomial output distribution induced by the reference input $X_r\sim\mathrm{Beta}(1/2,1/2)$ is asymptotically optimal: it approaches the capacity-achieving output distribution in relative entropy and, after a comparison step, in $\chi^2$ divergence. Third, we prove a quantitative $\chi^2$ approximation lower bound showing that this Beta-binomial output cannot be approximated too well by the output induced by a $K$-point input. Combining these ingredients forces the capacity-achieving input distribution to have at least order $\sqrt{n\log\log n}$ mass points.
\end{abstract}

\section{Introduction}
We consider the binomial channel with input-output law
\begin{equation}
P_{Y|X}(y|x)=\binom{n}{y}x^y(1-x)^{n-y}, \quad x \in [0,1],\, y \in \{0,\ldots,n\}.
\end{equation}
The channel capacity is
\begin{equation}\label{eq:capacity_opt_problem}
C(n)=\max_{P_X:\,X\in[0,1]} I(X;Y) = I(X^\ast; Y^\ast),
\end{equation}
where $P_{X^\ast}$ and $P_{Y^\ast}$ are the capacity-achieving input and output distributions, respectively.

The binomial channel naturally arises in molecular communications, and the interested reader is referred to
\cite{farsad2020capacities,einolghozati2013design,farsad2017capacity,jamali2019channel}
and references therein. The channel is also useful in the study of the deletion channel
\cite{levenshtein1966binary,cheraghchi2019capacity}. The capacity of the binomial channel was first considered in
\cite{komninakis2001capacity}, where the minimax redundancy theorem of \cite{xie1997minimax}
was used to show that the capacity scales asymptotically as $\frac12 \log n$. The exact capacity for the case
$n=1$ was computed in \cite{farsad2020capacities}, where the binary distribution supported on $\{0,1\}$ was shown to be capacity-achieving. However, sharp non-asymptotic bounds on the capacity remain limited.

The structure of the capacity-achieving input distribution has also been studied. For example, the authors of
\cite{farsad2020capacities} designed an algorithm for computing the capacity and a capacity-achieving input distribution using a dual representation of the maximization problem. It is also known, by the Witsenhausen technique
\cite{WitsenhausenBOund}, that there exists a capacity-achieving distribution with at most $n+1$ mass points.

More recently, \cite{Zieder2024} studied structural properties of capacity-achieving distributions for the binomial channel. In particular, it was shown that all capacity-achieving distributions are discrete. This discreteness was then used to establish strong concavity of the mutual information, which implies uniqueness of $P_{X^\star}$. It was also shown that the capacity-achieving distribution is symmetric around $\frac12$. The same work improved the Witsenhausen upper bound of $n+1$ mass points to an upper bound of order $\frac n2$, and proved a lower bound of order $\sqrt n$. The authors also derived capacity bounds. However, the gap between these bounds does not vanish with $n$, which limits their use for obtaining sharper cardinality lower bounds.

In this work, we improve the lower bound on the number of mass points of $P_{X^\ast}$ from order $\sqrt n$ to order $\sqrt{n\log\log n}$. We also derive new upper and lower bounds on the capacity whose gap vanishes as $n\to\infty$. This vanishing-gap estimate is a key ingredient in the support-size argument: it shows that the output distribution induced by a capacity-achieving input must be close to a natural Beta-binomial reference output.

The techniques used in this work are inspired by a recent result on the additive Gaussian channel with a peak-power constraint \cite{wang2026improvedlowerboundcardinality}, where improved lower bounds were obtained on the support size of the capacity-achieving input distribution. The main idea of \cite{wang2026improvedlowerboundcardinality} is to combine the golden formula with lower bounds from best-approximation theory, which quantify how well a target output distribution can be approximated by finite Gaussian mixtures~\cite{BestApporximationGaussianMixture}. In the present paper, we adapt this strategy to the binomial setting by replacing Gaussian-mixture approximation with approximation by finite binomial mixtures.

\subsection{Main Contributions and Outline}

The main contribution of this paper is a new lower bound on the support size of the capacity-achieving input distribution for the binomial channel. In particular, we improve the previously known lower bound of order $\sqrt n$ to a lower bound of order $\sqrt{n\log\log n}$, up to explicit constants; see Theorem~\ref{thm:main-support}.

The proof is based on several ingredients, which may also be of independent interest.
First, in Section ~\ref{sec:preliminaries} we discuss the preliminaries we need in our proofs.
Then, in Section~\ref{sec:main-results}, we derive new explicit upper and lower bounds on the capacity. The lower bound is obtained by evaluating the mutual information at the reference input $X_r\sim\mathrm{Beta}(1/2,1/2)$, while the upper bound follows from the Xie--Barron minimax redundancy construction; see Theorems~\ref{thm:capacity-lb} and~\ref{thm:capacity-ub}. These bounds imply
\begin{equation}
    C(n)=\frac12\log\!\left(\frac{n\pi}{2e}\right)+o(1),
\end{equation}
and, more importantly for our purposes, the gap between the lower and upper bounds vanishes as $n\to\infty$; see Proposition~\ref{prop:capacity_gap}.

Second, we show that the output distribution $P_{Y_r}$ induced by the reference input $P_{X_r}$ is close to the capacity-achieving output distribution. More precisely, in Section~\ref{sec:main-results}, we prove that the relative entropy $\kl{P_{Y_r}}{P_{Y^\ast}}$ vanishes at the same rate as the capacity gap, and then convert this estimate into a $\chi^2$ bound; see Theorem~\ref{thm:optimal-output} and Proposition~\ref{prop:chi2-output}. This step is what connects the capacity estimates to the support-size question.

Third, in Section~\ref{sec:binomial-mixture-approximation}, we use a finite-mixture approximation lower bound for binomial mixtures; see Theorem~\ref{thm:main}. This result shows that the reference Beta-binomial output $P_{Y_r}$ cannot be approximated too well, in $\chi^2$ divergence, by the output induced by an input with only $K$ mass points. Combining this approximation lower bound with the upper bound on $\chi^2(P_{Y^\ast}\|P_{Y_r})$ forces the capacity-achieving input to have sufficiently many mass points.

Finally, Section~\ref{sec:proofs} contains the proofs, and Section~\ref{sec:conclusion} concludes with several open questions. 

\subsection{Notation}
All logarithms are to base $e$. Deterministic quantities are denoted by lower-case letters and random variables by upper-case letters. For a random variable $X$ and every measurable subset $\cA \subseteq \bbR$ the probability distribution is written as $P_{X}(\cA) = \bbP[X \in \cA]$. The support set of $P_X$ is
\begin{equation}
\supp(P_{X})=\{x:  \text{ $P_{X}( \mathcal{D})>0$ for every open set $ \mathcal{D} \ni x $} \}. 
\end{equation} 
When $X$ is discrete, we write $P_X(x)$ for $P_X(\{x\})$, i.e., $P_X$ is a probability mass function (pmf). The relative entropy and the $\chi^2$ divergence between the distributions $P$ and $Q$ are denoted by $\kl{P}{Q}$ and $\chi^2(P\,\| \,Q)$, respectively.

\section{Preliminaries}\label{sec:preliminaries}

\subsection{KKT Conditions}
The next result  provides KKT conditions for the optimization problem in~\eqref{eq:capacity_opt_problem}, allowing the study of the support properties of an optimal input distribution (see for example~\cite{CISS_2018}).

\begin{lemma}\label{lem:KKT}
$P_{X^\ast}$ is a capacity-achieving input distribution if and only if the following conditions hold:
\begin{align}
\kl{P_{Y|X}(\cdot|x)}{P_{Y^\ast}} &\le C(n), \qquad x\in[0,1], \label{eq:KKT-ineq}\\
\kl{P_{Y|X}(\cdot|x)}{P_{Y^\ast}} &= C(n), \qquad x\in\supp(P_{X^\ast}), \label{eq:KKT-eq}
\end{align}
where  $P_{Y^\ast}$ is the induced output distribution.
\end{lemma}
\subsection{Conditional Expectation Operators}
Define the forward conditional expectation operator $A: \mathbb{R}^{\{0,\dots,n\}} \to \mathbb{R}[x]$ by
\begin{equation}
    (A g)(x) = \mathbb{E}[g(Y) | X=x] = \sum_{y=0}^n g(y) P_{Y|X}(y|x).
\end{equation}
Using the falling factorial $y^{\underline{k}} = y(y-1)\cdots(y-k+1)$, the moments of the binomial distribution yield
\begin{equation} \label{eq:A_action}
    A(y^{\underline{k}}) = n^{\underline{k}} x^k.
\end{equation}
Thus, $A$ maps polynomials in $y$ of degree $k$ to polynomials in $x$ of degree $k$.

Define the backward (posterior) conditional expectation operator $B: L^1(P_{X}) \to \mathbb{R}^{\{0,\dots,n\}}$ by
\begin{equation}
    (B f)(y) = \mathbb{E}[f(X) | Y=y],
\end{equation}
where the expectation is taken with respect to the posterior distribution $P_{X|Y}(x|y)$ induced by the prior $X$. By Bayes' rule, the posterior is proportional to $x^{y-1/2} (1-x)^{n-y-1/2}$, which is a $\text{Beta}(y+1/2, n-y+1/2)$ distribution. The $k$-th moment of this Beta distribution is
\begin{equation} \label{eq:B_action}
    B(x^k) = \frac{(y+1/2)_k}{(n+1)_k},
\end{equation}
where $(a)_k = a(a+1)\cdots(a+k-1)$ is the rising factorial. Thus, $B$ maps polynomials in $x$ of degree $k$ to polynomials in $y$ of degree $k$.

These operators satisfy the adjoint relation: for any $f \in L^1(P_{X})$ and $g \in \mathbb{R}^{\{0,\dots,n\}}$,
\begin{equation} \label{eq:adjoint}
    \mathbb{E}_{Y}[g(Y) (B f)(Y)] = \mathbb{E}_{X}[f(X) (A g)(X)].
\end{equation}
\subsection{Orthogonal Polynomials}
The orthogonal polynomials for the $\text{Beta}(1/2, 1/2)$ distribution are the shifted Chebyshev polynomials of the first kind, defined as $\tilde{T}_k(x) = \cos(k \arccos(2x-1))$. They satisfy the orthogonality condition:
\begin{equation}
    \mathbb{E}_{X}[\tilde{T}_k(X) \tilde{T}_m(X)] = 
    \begin{cases} 
      1 & \text{if } k = m = 0, \\
      \frac{1}{2} & \text{if } k = m \ge 1, \\
      0 & \text{if } k \neq m.
   \end{cases}
\end{equation}
Since $A$ is a bijection between polynomials of degree $k \le n$, we can uniquely define a polynomial $H_k(y)$ of degree $k$ such that
\begin{equation}
    (A H_k)(x) = \tilde{T}_k(x).
\end{equation}
We claim that $H_k(y)$ are orthogonal polynomials for $Y$. Let $m < k \le n$. Since $B$ is bijective on polynomials of degree up to $n$, there exists a polynomial $q_m(x)$ of degree $m$ such that $B q_m = y^m$. Using the adjoint relation \eqref{eq:adjoint}:
\begin{align}
    \mathbb{E}_{Y}[H_k(Y) Y^m] &= \mathbb{E}_{Y}[H_k(Y) (B q_m)(Y)] \nonumber \\
    &= \mathbb{E}_{X}[q_m(X) (A H_k)(X)] \nonumber \\
    &= \mathbb{E}_{X}[q_m(X) \tilde{T}_k(X)] = 0,
\end{align}
because $\tilde{T}_k(x)$ is orthogonal to all polynomials of degree strictly less than $k$.

\subsection{Squared Norms of $H_k(Y)$}
Let $h_k = \mathbb{E}_{Y}[H_k(Y)^2]$. Using the adjoint relation again:
\begin{align}
    h_k &= \mathbb{E}_{Y}[H_k(Y) (B B^{-1} H_k)(Y)] \nonumber \\
    &= \mathbb{E}_{X}[(B^{-1} H_k)(X) (A H_k)(X)] \nonumber \\
    &= \mathbb{E}_{X}[(B^{-1} H_k)(X) \tilde{T}_k(X)].
\end{align}
Since $\tilde{T}_k(x)$ is orthogonal to lower-degree terms, we only need the leading coefficient of $B^{-1} H_k(x)$. The leading term of $\tilde{T}_k(x)$ for $k \ge 1$ is $2^{2k-1} x^k$.
From \eqref{eq:A_action}, $A(y^k) = n^{\underline{k}} x^k + \text{lower order terms}$, which implies $A^{-1}(x^k) = \frac{1}{n^{\underline{k}}} y^k + \dots$
Thus,
\begin{equation}
    H_k(y) = A^{-1}(\tilde{T}_k)(y) = \frac{2^{2k-1}}{n^{\underline{k}}} y^k + \dots
\end{equation}
From \eqref{eq:B_action}, $B(x^k) = \frac{1}{(n+1)_k} y^k + \dots$, which implies $B^{-1}(y^k) = (n+1)_k x^k + \dots$
Composing these gives
\begin{equation}
    (B^{-1} H_k)(x) = \frac{2^{2k-1}}{n^{\underline{k}}} (n+1)_k x^k + \dots
\end{equation}
Since $x^k = 2^{-(2k-1)} \tilde{T}_k(x) + \dots$, we can write
\begin{equation}
    (B^{-1} H_k)(x) = \frac{(n+1)_k}{n^{\underline{k}}} \tilde{T}_k(x) + \dots
\end{equation}
Therefore, for $k \ge 1$:
\begin{equation} \label{eq:hk_exact}
    h_k = \frac{(n+1)_k}{n^{\underline{k}}} \mathbb{E}_{X}[\tilde{T}_k(X)^2] = \frac{1}{2} \prod_{j=1}^k \frac{n+j}{n-j+1}.
\end{equation}

\subsection{Parseval's Identity}
We expand the density ratio in the orthogonal basis $H_k(y)$:
\begin{equation}
    \frac{P_Y(y)}{P_{Y^*}(y)} = \sum_{k=0}^n a_k H_k(y).
\end{equation}
The coefficients are given by
\begin{align}
    a_k &= \frac{1}{h_k} \mathbb{E}_{Y^*}\left[ \frac{P_Y(Y^*)}{P_{Y^*}(Y^*)} H_k(Y^*) \right] = \frac{1}{h_k} \mathbb{E}_Y[H_k(Y)] \nonumber \\
    &= \frac{1}{h_k} \mathbb{E}_X[(A H_k)(X)] = \frac{1}{h_k} \mathbb{E}_X[\tilde{T}_k(X)].
\end{align}
Let $\epsilon_k = \mathbb{E}_X[\tilde{T}_k(X)]$. Note that $\epsilon_0 = 1$ and $h_0 = 1$. By Parseval's identity, the $\chi^2$ divergence is
\begin{align}\label{eq:chi_square_div_Parseval}
    \chi^2(P_Y \| P_{Y^*}) &= \mathbb{E}_{Y^*}\left[ \left( \frac{P_Y(Y^*)}{P_{Y^*}(Y^*)} - 1 \right)^2 \right] \nonumber \\
    &= \sum_{k=1}^n a_k^2 h_k = \sum_{k=1}^n \frac{\epsilon_k^2}{h_k}.
\end{align}

\section{Main Results}\label{sec:main-results}
We now summarize the main statements proved in this paper.

\subsection{Asymptotically Tight Capacity Bounds}
The next theorem provides a lower bound on the capacity by evaluating mutual information at the reference input $X_r\sim\mathrm{Beta}(1/2,1/2)$ with probability density function (pdf)
\begin{equation}\label{eq:pdf_beta}
    f_{X_r}(x) = \frac{1}{\pi \sqrt{x(1-x)}}, \quad x \in (0,1).
\end{equation}
We will also need the expression for the output distribution induced by $X_r$, which is the Beta-binomial distribution~\cite{JohnsonKempKotz2005}:
\begin{equation} \label{eq:output_dist_X_r}
P_{Y_r}(y)
=
\frac{\Gamma\!\left(y+\tfrac12\right)\Gamma\!\left(n-y+\tfrac12\right)}
{\pi\,\Gamma(y+1)\Gamma(n-y+1)},
\qquad  y=0,\dots,n.
\end{equation}

\begin{theorem}\label{thm:capacity-lb}
For every $n\ge 1$,
\begin{equation}\label{eq:capacity-lb}
C(n) \ge \psi(n+1)-\log\!\bigl(1+\sqrt{3n+1}\bigr)+\frac12\log\!\left(\frac{3\pi}{2e}\right),
\end{equation}
where $\psi(\cdot)$ denotes the digamma function. In particular,
\begin{equation}\label{eq:capacity-lb-asymp}
C(n) \ge  \frac12\log\!\left(\frac{n\pi}{2e}\right) + r_{{\rm LB}}(n),
\end{equation}
where
\begin{equation}
r_{{\rm LB}}(n) := \frac12\log\!\left(1+\frac{1}{3n}\right)-\log\!\left(1+\frac{1}{\sqrt{3n+1}}\right)-\frac{1}{n+1}.
\end{equation}

\end{theorem}
\begin{IEEEproof}
    The proof is given in Section~\ref{sec:proof_thm:capacity-lb}.
\end{IEEEproof}

The following upper bound is based on the Xie--Barron minimax redundancy construction~\cite{xie1997minimax}.

\begin{theorem}\label{thm:capacity-ub}
For all $n\ge 28$, we have
\begin{equation}\label{eq:capacity-ub}
C(n) \le \frac12\log\!\left(\frac{n\pi}{2e}\right) + r_{{\rm UB}}(n),
\end{equation}
where
\begin{equation}
r_{{\rm UB}}(n) := -\log\!\left(1-2\left(\frac{2e}{n\pi}\right)^{1/4}\right)
+\frac{10}{\log\frac{n\pi}{2e}}.
\end{equation}
\end{theorem}
\begin{IEEEproof}
    The proof is given in Section~\ref{sec:proof_thm:capacity-up}.
\end{IEEEproof}

Combining Theorems~\ref{thm:capacity-lb} and~\ref{thm:capacity-ub}, we arrive at the following bound on the gap between the upper and lower bounds.
\begin{prop}\label{prop:capacity_gap}
    The gap between the capacity upper and lower bounds of Theorems~\ref{thm:capacity-lb} and~\ref{thm:capacity-ub} is
    \begin{align}
        \mathrm{Gap}(n) :=r_{{\rm UB}}(n)-r_{{\rm LB}}(n) \le \frac{17}{\log\!\left(\frac{n\pi}{2e}\right)}
    \end{align}
    for $n \ge 444$. Therefore, the gap vanishes as $n\to \infty$.
\end{prop}
\begin{IEEEproof}
Recall that
    \begin{align}
{\rm Gap}(n) &= r_{{\rm UB}}(n)-r_{{\rm LB}}(n) \\
&=-\log\!\left(1-2\left(\frac{2e}{n\pi}\right)^{1/4}\right)
+\frac{10}{\log\!\left(\frac{n\pi}{2e}\right)} \nonumber\\
&\quad -\frac12\log\!\left(1+\frac{1}{3n}\right)
+\log\!\left(1+\frac{1}{\sqrt{3n+1}}\right)+\frac{1}{n+1}.
\end{align}

Using $-\log(1-x)\le x + x^2$ for $0<x<1/2$, and noting that for sufficiently large $n$ (e.g., $n \ge 444$)
\[
x = 2\left(\frac{2e}{n\pi}\right)^{1/4} < \frac12,
\]
we obtain
\[
-\log\!\left(1-2\left(\frac{2e}{n\pi}\right)^{1/4}\right)
\le
2\left(\frac{2e}{n\pi}\right)^{1/4}
+
4\left(\frac{2e}{n\pi}\right)^{1/2}.
\]

Next, using $\log(1+x)\le x$ for $x\ge0$,
\[
\log\!\left(1+\frac{1}{\sqrt{3n+1}}\right)
\le\frac{1}{\sqrt{3n+1}}.
\]
Moreover,
\[
-\log\!\left(1+\frac{1}{3n}\right)
\le 0.
\]

Therefore,
\[
\log\!\left(1+\frac{1}{\sqrt{3n+1}}\right) +\frac{1}{n+1}\le \frac{1}{\sqrt{3n+1}}+\frac{1}{n+1} \leq \frac{1}{\sqrt{n}} + \frac{1}{n}.
\]
\begin{align}
{\rm Gap}(n)
&\le
\frac{10}{\log\!\left(\frac{n\pi}{2e}\right)}
+
\left(\frac{32e}{\pi}\right)^{\frac{1}{4}}\left(\frac{1}{n}\right)^{\frac{1}{4}} \\
& \quad \ +\left(1+\left(\frac{32e}{\pi}\right)^{\frac{1}{2}}\right)\left(\frac{1}{n}\right)^{\frac{1}{2}} +\frac{1}{n}.
\end{align}

For all $n \ge 444$, the polynomial-decay terms are dominated by $n^{-1/4}$. Hence
\begin{align}
{\rm Gap}(n)
&\le
\frac{10}{\log\!\left(\frac{n\pi}{2e}\right)}
+
C_1 n^{-1/4} \\
&= \frac{10}{\log\!\left(\frac{n\pi}{2e}\right)}
+
5 n^{-1/4} \label{eq:pick_C1} \\
&\le \frac{17}{\log\!\left(\frac{n\pi}{2e}\right)}, \qquad n \ge 444, \label{eq:bound_Un}
\end{align}
where in~\eqref{eq:pick_C1} we picked $C_1 = 5$.
\end{IEEEproof}
As we will see in Section~\ref{sec:proof_lowerbound_supp}, characterizing the capacity up to a vanishing gap in $n$ is essential to establish the new lower bound on the support size of the capacity-achieving distribution $P_{X^\ast}$.

\subsection{Asymptotically Optimal Output Distribution}
In Proposition~\ref{prop:capacity_gap}, we have shown that the input $X_r$ is asymptotically optimal from the rate perspective. We now show that the output distribution induced by $X_r$ also approaches the optimal output distribution.

The next result makes this connection precise.
\begin{theorem}\label{thm:optimal-output}
Let $P_{Y^\ast}$ be the output distribution induced by the capacity-achieving input $P_{X^\ast}$, and let $P_{Y_r}$ be defined as in \eqref{eq:output_dist_X_r}. Then, the following statements hold:
\begin{itemize}
    \item There exists an absolute constant $c_\star>0$ such that, for every $n\ge 1$ and every
$y\in\{0,\dots,n\}$,
\begin{equation}
\frac{P_{Y_r}(y)}{P_{Y^\ast}(y)} \ge c_\star, \label{eq:lower_bound_ratio_output_dist}
\end{equation}
where one can take 
$c_\star = \frac{e^{-8/25}}{\sqrt{6e/\pi}\left(\frac{5\pi}{2}+6\right)}
            \approx 0.02300399$;      
and 
\item For $n \ge 444$,
\begin{equation}
    \kl{P_{Y_r}}{P_{Y^{\ast}}} 
    \le \frac{17}{\log\!\left(\frac{n\pi}{2e}\right)}. \label{eq:upper_gap}
\end{equation}
\end{itemize}
\end{theorem}
\begin{IEEEproof} 
   The proof is given in Section~\ref{proofoftheorem3}.
\end{IEEEproof}

For our purposes, it will be easier to work with $\chi^2$-divergence. To that end, we have the following result.

\begin{prop}\label{prop:chi2-output}
For $n \ge 444$,
    \begin{equation}
    \chi^2(P_{Y^\ast} \| P_{Y_r}) \le \zeta(c_\star^{-1}) \frac{17}{\log\!\left(\frac{n\pi}{2e}\right)},
    \label{eq:upper_chi2}
\end{equation}
where 
\begin{equation}
    \zeta(t) = \frac{(t-1)^2}{t-1-\log(t)}.
\end{equation}
\end{prop}
\begin{IEEEproof}
    By combining \cite[Eq.~(169)]{sason2016f} and \cite[Thm.~6]{sason2016f}, we obtain
\begin{equation}\label{eq:from_chi2_to_kl}
    \chi^2(P_{Y^\ast} \| P_{Y_r}) \le \zeta(\beta_1^{-1}) \kl{P_{Y_r}}{P_{Y^{\ast}}},
\end{equation}
where
\begin{equation}\label{eq:lower_beta1}
    \beta_1 := \inf_y \frac{P_{Y_r}}{P_{Y^\ast}}(y) \ge c_\star >0.
\end{equation}
Putting together~\eqref{eq:from_chi2_to_kl} and~\eqref{eq:upper_gap} yields
\begin{align}
    \chi^2(P_{Y^\ast} \| P_{Y_r}) &\le \zeta(\beta_1^{-1}) \frac{17}{\log\!\left(\frac{n\pi}{2e}\right)} \\
    &\le \zeta(c_\star^{-1}) \frac{17}{\log\!\left(\frac{n\pi}{2e}\right)},
\end{align}
where in the last step we used~\eqref{eq:lower_beta1} and that $t \mapsto \zeta(t)$ is increasing for $t\ge 1$.
\end{IEEEproof}

\subsection{A Lower Bound on the Support Size}

We now state the main result of this work, which is an improved lower bound on the cardinality of $P_{X^\ast}$.

\begin{theorem}\label{thm:main-support}
For every $n \ge 1$,
\begin{equation}\label{eq:explicit-support}
|\supp(P_{X^\ast})| \ge \max  \left\{2, \,  \rme^{C(n)}, \,  \frac{1}{8}
\sqrt{
n \,
\log^{+}\!\left(
\frac{\log\!\left(\frac{n\pi}{2e}\right)}{3268 }
\right) 
} \right \},
\end{equation}
where $C(n)$ is the capacity defined in \eqref{eq:capacity_opt_problem} and $\log^+(x)=\max\{0,\log(x)\}$.
\end{theorem}
\begin{IEEEproof}
 See Section~\ref{sec:proof_lowerbound_supp}. 
\end{IEEEproof}

We make the following remarks:
\begin{itemize}
    \item  The first two bounds in \eqref{eq:explicit-support} have been shown in \cite{zieder2024binomial}.
    \item In view of Theorems~\ref{thm:capacity-lb} and~\ref{thm:capacity-ub}, we have that  $\rme^{C(n)} =\Theta(\sqrt{n})$. Therefore, for large $n$ the bound in \eqref{eq:explicit-support} can be expressed as 
    \begin{equation}
        |\supp(P_{X^\ast})| \ge \Omega\!\left(\sqrt{n\log\log n}\right).
    \end{equation}
    \item The explicit constants are not optimized. In particular, the improvement over the order-$\sqrt n$ lower bound is asymptotic and becomes visible only for sufficiently large $n$.
\end{itemize}

\section{On the Best Approximation Theory of Finite Binomial Mixtures}
\label{sec:binomial-mixture-approximation}

The key result for providing a lower bound on the support size is the following theorem. It quantifies the best possible approximation of the reference Beta-binomial distribution $P_{Y_r}$ by a binomial mixture with finitely many components. Its proof follows from an adaptation of the trigonometric moment method for mixture of Gaussian distributions of \cite[Thm.~7]{BestApporximationGaussianMixture}.

\begin{theorem} \label{thm:main}
Let $X$ be a discrete random variable with $K$ mass points in $[0,1]$. Let $Y$ be the output of the binomial channel induced by $X$. Let $Y_r$ be the output distribution induced by the reference input $X_r \sim \text{Beta}(1/2, 1/2)$. Then, for any integer $L$ such that $K < L \le \frac{n+2}{2}$, we have
\begin{equation}
    \chi^2(P_Y \| P_{Y_r}) \ge B_n(L) := \frac{L-K}{2L \prod_{j=1}^{2L-2} \frac{n+j}{n-j+1}}.
    \label{lowerchi}
\end{equation}
\end{theorem}
\begin{IEEEproof}
The proof proceeds by leveraging the spectral properties of the conditional expectation operators associated with the binomial channel and applying the Eckart-Young-Mirsky theorem to a moment matrix of arbitrary dimension~$L$.

For any integer $m \le n$, we can truncate the sum to lower bound the divergence \eqref{eq:chi_square_div_Parseval}:
\begin{equation} \label{eq:chi2_trunc}
    \chi^2(P_Y \| P_{Y_r}) \ge \sum_{k=1}^{m} \frac{\epsilon_k^2}{h_k}.
\end{equation}

Let $L$ be an integer such that $K < L$ and $2L-2 \le n$. Define the $L \times L$ moment matrix $M$ with entries $M_{i,j} = \mathbb{E}_X[\tilde{T}_i(X) \tilde{T}_j(X)]$ for $0 \le i, j \le L-1$. Since $X$ is supported on $K$ mass points, $M$ can be written as the sum of $K$ rank-1 matrices, implying $\text{rank}(M) \le K$.

Let $D$ be the corresponding matrix for the reference distribution $X_r$, so $D_{i,j} = \mathbb{E}_{X_r}[\tilde{T}_i(X_r) \tilde{T}_j(X_r)]$. The matrix $D$ is diagonal with entries $D_{0,0} = 1$ and $D_{i,i} = 1/2$ for $i \ge 1$. The singular values of $D$ are $1$ (with multiplicity 1) and $1/2$ (with multiplicity $L-1$). 

By the Eckart-Young-Mirsky theorem \cite{golub1988generalization}, the Frobenius distance from $D$ to any matrix of rank at most $K$ is lower-bounded by the sum of the squared $L-K$ smallest singular values of $D$. Since $K \ge 1$, the $L-K$ smallest singular values are all $1/2$. Thus,
\begin{equation} \label{eq:EYM}
    \|M - D\|_F^2 \ge (L-K) \left(\frac{1}{2}\right)^2 = \frac{L-K}{4}.
\end{equation}

Using the trigonometric identity $\cos(i\theta)\cos(j\theta) = \frac{1}{2}(\cos((i+j)\theta) + \cos(|i-j|\theta))$, the product of Chebyshev polynomials linearizes as
\begin{equation}
    \tilde{T}_i(x) \tilde{T}_j(x) = \frac{1}{2} \left( \tilde{T}_{i+j}(x) + \tilde{T}_{|i-j|}(x) \right).
\end{equation}
Thus, the entries of the moment matrix $M$ are given by $M_{i,j} = \frac{1}{2} (\epsilon_{i+j} + \epsilon_{|i-j|})$.
To elegantly express the difference matrix $M - D$, we define a modified sequence $\tilde{\epsilon}_k$:
\begin{equation}
    \tilde{\epsilon}_k = \begin{cases} \epsilon_k & \text{if } k \ge 1, \\ 0 & \text{if } k = 0. \end{cases}
\end{equation}
Since $\epsilon_0 = 1$, it is straightforward to verify that for all $0 \le i, j \le L-1$,
\begin{equation}
    M_{i,j} - D_{i,j} = \frac{1}{2} (\tilde{\epsilon}_{i+j} + \tilde{\epsilon}_{|i-j|}).
\end{equation}
Indeed, for $i \neq j$, $D_{i,j} = 0$ and the indices $i+j, |i-j| \ge 1$. For $i = j \ge 1$, $D_{i,i} = 1/2$ and $\frac{1}{2}(\tilde{\epsilon}_{2i} + \tilde{\epsilon}_0) = \frac{1}{2}\epsilon_{2i} = M_{i,i} - 1/2$. For $i=j=0$, $D_{0,0} = 1$ and $\frac{1}{2}(\tilde{\epsilon}_0 + \tilde{\epsilon}_0) = 0 = M_{0,0} - 1$.

Using the inequality $(a+b)^2 \le 2a^2 + 2b^2$, we bound the squared entries:
\begin{equation}
    (M_{i,j} - D_{i,j})^2 \le \frac{1}{2} (\tilde{\epsilon}_{i+j}^2 + \tilde{\epsilon}_{|i-j|}^2).
\end{equation}
Summing over all $0 \le i, j \le L-1$:
\begin{equation}
    \|M - D\|_F^2 \le \frac{1}{2} \sum_{i,j=0}^{L-1} (\tilde{\epsilon}_{i+j}^2 + \tilde{\epsilon}_{|i-j|}^2) = \sum_{k=1}^{2L-2} S_k \tilde{\epsilon}_k^2,
\end{equation}
where $S_k = \frac{1}{2} (N_1(k) + N_2(k))$, with $N_1(k)$ being the number of pairs $(i,j) \in \{0,\dots,L-1\}^2$ such that $i+j=k$, and $N_2(k)$ the number of pairs such that $|i-j|=k$. The $k=0$ term vanishes because $\tilde{\epsilon}_0 = 0$.
For $1 \le k \le L-1$, $N_1(k) = k+1$ and $N_2(k) = 2(L-k)$, giving $S_k = L - \frac{k-1}{2} \le L$.
For $L \le k \le 2L-2$, $N_1(k) = 2L - 1 - k$ and $N_2(k) = 0$, giving $S_k = L - \frac{k+1}{2} \le L$.
Since $\tilde{\epsilon}_k = \epsilon_k$ for $k \ge 1$, we obtain:
\begin{equation}
    \|M - D\|_F^2 \le L \sum_{k=1}^{2L-2} \epsilon_k^2.
\end{equation}
Combining this with \eqref{eq:EYM} yields:
\begin{equation} \label{eq:epsilon_bound}
    \sum_{k=1}^{2L-2} \epsilon_k^2 \ge \frac{L-K}{4L}.
\end{equation}

Notice from \eqref{eq:hk_exact} that $h_k$ is strictly increasing with $k$ because each factor $\frac{n+j}{n-j+1} > 1$ for $j \ge 1$. Thus, $h_k \le h_{2L-2}$ for all $1 \le k \le 2L-2$. Returning to the $\chi^2$ divergence \eqref{eq:chi2_trunc} with $m = 2L-2$:
\begin{equation}
    \chi^2(P_Y \| P_{Y_r}) \ge \sum_{k=1}^{2L-2} \frac{\epsilon_k^2}{h_k} \ge \frac{1}{h_{2L-2}} \sum_{k=1}^{2L-2} \epsilon_k^2 \ge \frac{L-K}{4L h_{2L-2}}.
\end{equation}
Substituting the expression for $h_{2L-2}$ yields the final lower bound:
\begin{equation}
\label{chisq}
    \chi^2(P_Y \| P_{Y_r}) \ge \frac{L-K}{2L \prod_{j=1}^{2L-2} \frac{n+j}{n-j+1}}.
\end{equation}
This completes the proof of Theorem \ref{thm:main}.
\end{IEEEproof}

The lower bound in Theorem~\ref{thm:main} depends on the free integer parameter $L$ satisfying $K < L$ and $2L-2 \le n$. To obtain the strongest possible bound, we should maximize $B_n(L)$ over all admissible values of $L$. This optimization does not have a closed-form solution. Thus, we establish an explicit lower bound in the next proposition.
\begin{prop}\label{prop:uniform-explicit}
Let $L$ be an integer such that $K < L \le \frac{n+2}{2}$. Then
\begin{equation}
B_{n}(L)
\ge
\frac{L-K}{2L}\exp\!\left(-\frac{(2L-2)^2}{\,n-2L+3\,}\right).
\end{equation}
\end{prop}
\begin{IEEEproof}
Set $m:=2L-2$. Then
\begin{align}
\log\!\left(\prod_{j=1}^{m}\frac{n+j}{n-j+1}\right)
&=
\sum_{j=1}^{m}\log\!\left(1+\frac{2j-1}{n-j+1}\right) \\
&\le
\sum_{j=1}^{m}\frac{2j-1}{n-j+1} \label{eq:use_logineq} \\
&\le
\frac{1}{n-m+1}\sum_{j=1}^{m}(2j-1) \\
&=\frac{m^2}{n-m+1}
\end{align}
where in \eqref{eq:use_logineq} we used $\log(1+x)\le x$ for $x>-1$. Substituting into $B_n(L)$, we have the final result
\[
B_{n}(L)
\ge
\frac{L-K}{2L}\exp\!\left(-\frac{m^2}{n-m+1}\right).
\] 

\end{IEEEproof}

\section{Proofs}
\label{sec:proofs}
This section collects some of the proofs. 

\subsection{Proof of Theorem~\ref{thm:capacity-lb}}
\label{sec:proof_thm:capacity-lb}
\begin{IEEEproof}
Let $X \sim \mathrm{Beta}\!\left(\frac{1}{2}, \frac{1}{2}\right)$, with pdf as in \eqref{eq:pdf_beta}.
Then,
\begin{equation}
C(n) \ge I(X;Y) = H(Y) - H(Y|X).
\end{equation}
We bound each term separately. Write
\begin{align}
    H(Y|X) &\le \frac{1}{2}\log(2\pi e) + \frac{1}{2}\mathbb{E}\big[\log(nX(1-X)+\tfrac{1}{12})\big] \label{eq:upper_HY_given_x} \\
    &= \frac{1}{2}\log(2\pi e)
+
\log\!\left(\frac{1+\sqrt{3n+1}}{4\sqrt{3}}\right),
\end{align}
where \eqref{eq:upper_HY_given_x} follows from \cite[Lemma 10]{zieder2024binomial}, and the last step follows from a closed form evaluation of the expectation by using the substitution $x=\sin^2(t/2)$ and the identity \cite{Ryzhik}
\begin{equation}
\int_0^{\pi/2} \log(a + b \sin^2 t)\,dt
=
\pi \log\!\left(\frac{\sqrt{a} + \sqrt{a+b}}{2}\right),
\quad a,b>0.
\end{equation}

Under this input, $Y$ follows a Beta-binomial distribution:
\begin{align}
P_Y(y) &=
\frac{\Gamma(y+\tfrac{1}{2}) \Gamma(n-y+\tfrac{1}{2})}
{\pi \Gamma(y+1)\Gamma(n-y+1)}    \\
&\le \frac{1}{\pi \sqrt{(y+\tfrac{1}{4})(n-y+\tfrac{1}{4})}},
\end{align}
where the last step follows from Kershaw's inequality $\frac{\Gamma(y+\tfrac{1}{2})}{\Gamma(y+1)}
\le (y+\tfrac{1}{4})^{-1/2}$ \cite{kershaw1983some}. Therefore,
\begin{align}
H(Y)
&\ge \log \pi + \mathbb{E}\left[\log\!\big(Y+\tfrac{1}{4}\big)\right] \\
&\ge \log \pi + \mathbb{E}\big[\psi(Y+\tfrac{1}{2})\big] \label{eq:log_psi} \\
&=\log\!\left(\frac{\pi}{4}\right)+\psi(n+1),  \label{eq:betabin_identity}
\end{align}
where in \eqref{eq:log_psi} we used 
$
\log\!\big(y+\tfrac{1}{4}\big)
\ge \psi\!\big(y+\tfrac{1}{2}\big)$, where $\psi$ is the digamma function; and in \eqref{eq:betabin_identity} we used the Beta-binomial identity
\begin{align}
\mathbb{E}\big[\psi(Y+\alpha)\big]
&=
\psi(n+\alpha+\beta) + \psi(\alpha) - \psi(\alpha+\beta) \\
&=\psi(n+1) - 2\log 2,
\end{align}
with $\alpha=\beta=\tfrac{1}{2}$.

Combining the bounds on the entropy terms,
we obtain claim \eqref{eq:capacity-lb}.

Using $\psi(n+1)\ge \log (n+\frac{1}{3} ) -\tfrac{1}{n+1}$, we obtain \eqref{eq:capacity-lb-asymp} where
\begin{equation}
r_{{\rm LB},n} =
\frac{1}{2}\log\!\left(1+\frac{1}{3n}\right)
-
\log\!\left(1+\frac{1}{\sqrt{3n+1}}\right)-\frac{1}{n+1}
\to 0.
\end{equation}
This completes the proof.
\end{IEEEproof}

\subsection{Proof of Theorem~\ref{thm:capacity-ub}}
\label{sec:proof_thm:capacity-up}
\begin{IEEEproof}
By the dual characterization of capacity,
\begin{equation} \label{eq:dual_capacity}
    C(n) \le \sup_{x \in[0,1]} \KL(P_x^{(n)} \| Q_n)
\end{equation}
for any output distribution $Q_n$. Consider the output distribution
\begin{align}\label{eq:Q_n}
    Q_n(y) &=  \binom{n}{y} \left( \eta_n \left(\frac{c_n}{n}\right)^y  \left(  1-\frac{c_n}{n}\right)^{n-y} \right.\nonumber \\
    &\quad +\eta_n \left(1-\frac{c_n}{n}\right)^y  \left(  \frac{c_n}{n}\right)^{n-y} \nonumber\\
    &\quad \left. +(1-2\eta_n) \int_0^1 x^y(1-x)^{n-y} f_{X_r}(x) \mathrm{d}x \right),
\end{align}
where $f_{X_r}$ is as in \eqref{eq:pdf_beta} and
\begin{equation}
    \eta_n \ge \left(\frac{2 \mathrm{e}}{n\pi} \right)^{\frac{1}{4}}, \quad 
    c_n \le \frac{1}{2}\log\frac{n\pi}{2 \mathrm{e}}.
\end{equation}

Xie and Barron \cite{xie1997minimax} have proven the following result.
\begin{theorem}\label{thm:Barron} For the output distribution $Q_n$ given in \eqref{eq:Q_n} and for all sufficiently large $n$ (e.g., pick $n\ge 28$) we have
\begin{equation} \label{eq:XieBarron_upper}
    \max_{x \in [0,1]} \KL(P_x^{(n)} \| Q_n) \le \frac{1}{2}\log\frac{n}{2\pi \mathrm{e}} +\log\frac{\pi}{1-2\eta_n} +\frac{5}{c_n}.
\end{equation}

\end{theorem}
\begin{IEEEproof}
    See \cite[Sec.~III.B]{xie1997minimax}.
\end{IEEEproof}

Note that for $n\ge 28$, we have $\eta_n<1/2$, so the quantity $\log(\pi/(1-2\eta_n))$ in the right-hand side of \eqref{eq:XieBarron_upper} is well defined.
Using Theorem~\ref{thm:Barron} with \eqref{eq:dual_capacity} concludes the proof.
\end{IEEEproof}
{
\subsection{Proof of Theorem~\ref{thm:optimal-output}}
\label{proofoftheorem3}


Let
\begin{equation}
F_r(x)
=
\frac{2}{\pi}\arcsin\sqrt{x},
\qquad x\in[0,1],
\end{equation}
denote the cumulative distribution function of
$X_r\sim\mathrm{Beta}(1/2,1/2)$.

\begin{lemma}\label{lem:partition-mass}
Let
\begin{equation}
N
=
2\left\lceil \frac{5\pi}{4}\sqrt{n}\right\rceil
\end{equation}
and define
\begin{equation}
a_j
=
F_r^{-1}\!\left(\frac{j}{N}\right)
=
\sin^2\!\left(\frac{j\pi}{2N}\right),
\qquad
j=0,\ldots,N.
\end{equation}
Consider the partition of $[0,1]$
\begin{align}
A_j&=[a_j,a_{j+1}),
&&j=0,\ldots,N-2,\\
A_{N-1}&=[a_{N-1},1].
\end{align}
Then
\begin{equation}\label{eq:partition-mass}
P_{X^\ast}(A_j)
\le
\exp\!\left(\frac{8}{25}-C(n)\right),
\qquad
j=0,\ldots,N-1.
\end{equation}
\end{lemma}

\begin{proof}
We first record the set-mass bound used below. Let $A\subseteq[0,1]$
be measurable with $P_{X^\ast}(A)>0$, and choose
$x_A\in A\cap\operatorname{supp}(P_{X^\ast})$. By the KKT equality,
conditioning $P_{Y^\ast}$ on the event $\{X^\ast\in A\}$, the
log-sum inequality, and convexity of relative entropy in its second
argument,
\begin{equation}\label{eq:set-mass-general}
P_{X^\ast}(A)
\le
\exp\!\left(
n\sup_{x\in A}d(x_A\|x)-C(n)
\right).
\end{equation}

By construction,
\begin{equation}
F_r(a_j)=\frac{j}{N},
\end{equation}
and hence, for every $x,x'\in A_j$,
\begin{equation}\label{eq:Fr-cell}
|F_r(x)-F_r(x')|
\le
\frac{1}{N}.
\end{equation}
Moreover, $N$ is even and therefore $a_{N/2}=1/2$, so every cell is
contained either in $[0,1/2]$ or in $[1/2,1]$.

Fix $j$. If $P_{X^\ast}(A_j)=0$, there is nothing to prove.
Otherwise, choose $x_j\in A_j\cap\operatorname{supp}(P_{X^\ast})$;
for the first and last cells we take $x_0=0$ and $x_{N-1}=1$, using
$\{0,1\}\subseteq\operatorname{supp}(P_{X^\ast})$.

Suppose first that $A_j\subseteq[0,1/2]$ and write
\[
\theta_j=\arcsin\sqrt{x_j},
\qquad
\theta=\arcsin\sqrt{x}.
\]
Then $\theta_j,\theta\in[0,\pi/4]$. Applying
\cite[Lem.~3]{rooij2009learning} with
$\phi_1=\phi_3=\theta_j$ and $\phi_2=\theta$ gives
\begin{equation}
    d(x_j\|x)
    \leq
    \begin{cases}
        4(\theta-\theta_j)^2,
        & \theta_j\leq\theta,\\[1mm]
        4(\theta_j-\theta)^2
        \dfrac{\theta_j}{\theta},
        & \theta_j>\theta.
    \end{cases}
    \label{eq:rooij_cell_bound}
\end{equation}
When $\theta=0$, we necessarily have $j=0$ and
$x_j=x=0$, and hence $d(x_j\|x)=0$. Otherwise, if
$\theta_j>\theta$, then $j\geq1$ and, since $\theta_j$ and
$\theta$ belong to the same angular cell,
\begin{align}
    \frac{\theta_j}{\theta}
    &\leq
    \frac{(j+1)\pi/(2N)}
         {j\pi/(2N)}=
    \frac{j+1}{j}\leq
    2.
       \label{eq:angle_ratio_bound}
\end{align}
Consequently, for every $x\in A_j$,
\begin{align}
d(x_j\|x)
&\le
8(\theta_j-\theta)^2\\
&=
2\pi^2\bigl(F_r(x_j)-F_r(x)\bigr)^2\\
&\le
\frac{2\pi^2}{N^2},
\end{align}
where the last step follows from \eqref{eq:Fr-cell}.

If instead $A_j\subseteq[1/2,1]$, the same conclusion follows from
\[
d(x_j\|x)=d(1-x_j\|1-x)
\]
and
\[
F_r(1-x)=1-F_r(x).
\]
Thus, for every $j$ with $P_{X^\ast}(A_j)>0$,
\begin{equation}
n\sup_{x\in A_j}d(x_j\|x)
\le
\frac{2\pi^2n}{N^2}
\le
\frac{8}{25},
\end{equation}
because
\[
N\ge\frac{5\pi}{2}\sqrt{n}.
\]
Applying \eqref{eq:set-mass-general} with $A=A_j$ and $x_A=x_j$
proves \eqref{eq:partition-mass}.
\end{proof}

\begin{lemma}\label{lem:upper-sum}
Let $\{A_j\}_{j=0}^{N-1}$ be the partition in
Lemma~\ref{lem:partition-mass}. Then, for every
$y\in\{0,\ldots,n\}$,
\begin{equation}\label{eq:upper-sum}
\sum_{j=0}^{N-1}
\sup_{x\in A_j}P_{Y|X}(y|x)
\le
\left(\frac{5\pi}{2}+6\right)
\sqrt{n}\,P_{Y_r}(y).
\end{equation}
\end{lemma}

\begin{proof}
For fixed $y$, define
\begin{equation}
f_y(t)
=
P_{Y|X}\!\left(y\middle|F_r^{-1}(t)\right),
\qquad t\in[0,1],
\end{equation}
and
\begin{equation}
M_n(y)
=
\sup_{t\in[0,1]}f_y(t)
=
\sup_{x\in[0,1]}P_{Y|X}(y|x).
\end{equation}
Since $x\mapsto P_{Y|X}(y|x)$ is unimodal and $F_r^{-1}$ is
increasing, $f_y$ is unimodal. Therefore, its upper Darboux sum over
the uniform partition of $[0,1]$ satisfies
\begin{align}
\sum_{j=0}^{N-1}
\sup_{x\in A_j}P_{Y|X}(y|x)
&\le
N\int_0^1 f_y(t)\,dt+2M_n(y)\\
&=
N P_{Y_r}(y)+2M_n(y),
\label{eq:darboux}
\end{align}
where the last equality follows because $F_r(X_r)$ is uniform on
$[0,1]$.

We next show that
\begin{equation}\label{eq:Mn-bound}
M_n(y)
\le
2\sqrt{n}\,P_{Y_r}(y),
\qquad
y=0,\ldots,n.
\end{equation}
For $1\le y\le n-1$, the maximum is attained at $x=y/n$, and
Stirling's bound gives
\[
M_n(y)
\le
\sqrt{\frac{n}{2\pi y(n-y)}}.
\]
On the other hand, from the beta-binomial expression,
\[
P_{Y_r}(y)
\ge
\frac{1}
{\pi\sqrt{(y+\frac12)(n-y+\frac12)}}.
\]
Hence
\begin{align}
\frac{M_n(y)}{P_{Y_r}(y)}
&\le
\sqrt{
\frac{\pi n}{2}
\left(1+\frac{1}{2y}\right)
\left(1+\frac{1}{2(n-y)}\right)
}\\
&\le
\sqrt{\frac{9\pi n}{8}}
\le
2\sqrt{n}.
\end{align}
For $y=0$, $M_n(0)=1$ and
\[
P_{Y_r}(0)
=
\frac{\Gamma(n+\frac12)}
{\sqrt{\pi}\,\Gamma(n+1)}
\ge
\frac{1}{\sqrt{\pi(n+\frac12)}}
\ge
\frac{1}{2\sqrt{n}},
\]
with the case $n=1$ following directly from $P_{Y_r}(0)=1/2$.
The case $y=n$ follows by symmetry. This proves
\eqref{eq:Mn-bound}.

Finally,
\[
N
=
2\left\lceil\frac{5\pi}{4}\sqrt{n}\right\rceil
\le
\frac{5\pi}{2}\sqrt{n}+2.
\]
Substituting \eqref{eq:Mn-bound} into \eqref{eq:darboux}, and using
$2\le2\sqrt{n}$ for $n\ge1$, gives \eqref{eq:upper-sum}.
\end{proof}

\begin{proof}[Proof of Theorem~3]
Fix $y\in\{0,\ldots,n\}$. Using the partition in
Lemma~\ref{lem:partition-mass},
\begin{align}
P_{Y^\ast}(y)
&=
\sum_{j=0}^{N-1}
\int_{A_j}P_{Y|X}(y|x)\,dP_{X^\ast}(x)\\
&\le
\sum_{j=0}^{N-1}
P_{X^\ast}(A_j)
\sup_{x\in A_j}P_{Y|X}(y|x)\\
&\le
\exp\!\left(\frac{8}{25}-C(n)\right)
\sum_{j=0}^{N-1}
\sup_{x\in A_j}P_{Y|X}(y|x).
\end{align}
Applying Lemma~\ref{lem:upper-sum},
\begin{equation}\label{eq:ratio-pre}
\frac{P_{Y^\ast}(y)}{P_{Y_r}(y)}
\le
\exp\!\left(\frac{8}{25}-C(n)\right)
\left(\frac{5\pi}{2}+6\right)\sqrt{n}.
\end{equation}

From Theorem~1, together with
$\psi(n+1)\ge\log n$ and
$1+\sqrt{3n+1}\le3\sqrt{n}$, we have
\begin{equation}
C(n)
\ge
\frac12\log\!\left(\frac{n\pi}{6e}\right),
\end{equation}
and therefore
\begin{equation}
e^{-C(n)}\sqrt{n}
\le
\sqrt{\frac{6e}{\pi}}.
\end{equation}
Substituting this into \eqref{eq:ratio-pre} gives
\begin{equation}
\frac{P_{Y^\ast}(y)}{P_{Y_r}(y)}
\le
e^{8/25}
\sqrt{\frac{6e}{\pi}}
\left(\frac{5\pi}{2}+6\right)
=
c_\star^{-1},
\end{equation}
where
\begin{equation}
c_\star
=
\frac{e^{-8/25}}
{\sqrt{6e/\pi}\left(\frac{5\pi}{2}+6\right)}
\approx
0.02300399.
\end{equation}
Equivalently,
\begin{equation}
\frac{P_{Y_r}(y)}{P_{Y^\ast}(y)}
\ge
c_\star,
\qquad
y=0,\ldots,n.
\end{equation}

For the relative-entropy bound, the golden formula gives
\begin{equation}
D(P_{Y_r}\|P_{Y^\ast})
\le
C(n)-I(X_r;Y_r).
\end{equation}
By the capacity bounds of Theorems~1 and~2,
\begin{equation}
C(n)-I(X_r;Y_r)
\le
\operatorname{Gap}(n).
\end{equation}
Consequently, Proposition~1 yields
\begin{equation}
D(P_{Y_r}\|P_{Y^\ast})
\le
\operatorname{Gap}(n)
\le
\frac{17}
{\log\!\left(\frac{n\pi}{2e}\right)},
\qquad n\ge444.
\end{equation}
This completes the proof.
\end{proof}

}

\subsection{Proof of Theorem~\ref{thm:main-support}} \label{sec:proof_lowerbound_supp}

Before showing the proof of Theorem~\ref{thm:main-support}, we state the following helper proposition. 

\begin{prop}\label{prop:invert-chi2}
Let $X$ be a discrete random variable with $K$ mass points in $[0,1]$. Let $Y$ be the output of the binomial channel induced by $X$. Let $Y_r$ be the output induced by the reference input $X_r \sim \text{Beta}(1/2, 1/2)$. Assume that
\begin{equation} \label{eq:assumption_chi2}
\chi^2(P_Y\|P_{Y_r}) \le u_n
\end{equation}
for some $u_n\in(0,1/4)$, and define
\begin{equation}
\alpha_n := \log\frac{1}{4u_n}.
\end{equation}
Then at least one of the following holds:
\begin{equation}
K > \frac{n+2}{4},
\end{equation}
or
\begin{equation}\label{eq:invert-chi2}
K \ge \frac{4-\alpha_n+\sqrt{\alpha_n(4n+\alpha_n+4)}}{8}.
\end{equation}
In particular, if $\alpha_n\to\infty$, then $K/\sqrt n\to\infty$. More precisely, if moreover $\alpha_n=o(n)$, then
\begin{equation}
K\ge \frac14\sqrt{n\alpha_n}\,(1+o(1))
=
\frac14\sqrt{\,n\log\frac1{4u_n}\,}\,(1+o(1)).
\end{equation}
\end{prop}
\begin{IEEEproof}   
    If $K>(n+2)/4$, there is nothing to prove. Assume instead that
\begin{equation}
K\le \frac{n+2}{4}.
\end{equation}
Then $L=2K$ is admissible in Theorem~\ref{thm:main}, because
\begin{equation}
2L-2=4K-2\le n.
\end{equation}
Applying Proposition~\ref{prop:uniform-explicit} with $L=2K$ gives
\begin{equation}
\chi^2(P_Y\|P_{Y_r}) \ge \frac14\exp\!\left(-\frac{(4K-2)^2}{n-4K+3}\right).
\end{equation}
Combining this with the upper bound $\chi^2(P_Y\|P_{Y_r})\le u_n$, we obtain
\begin{equation}
u_n \ge \frac14\exp\!\left(-\frac{(4K-2)^2}{n-4K+3}\right),
\end{equation}
and therefore
\begin{equation}
\frac{(4K-2)^2}{n-4K+3}\ge \alpha_n,
\qquad
\alpha_n:=\log\frac{1}{4u_n}.
\end{equation}
Rearranging yields
\begin{equation}
16K^2+(4\alpha_n-16)K+(4-3\alpha_n-\alpha_n n)\ge 0.
\end{equation}
Solving this quadratic inequality for $K$ gives \eqref{eq:invert-chi2}.

Finally, if $\alpha_n\to\infty$, then either $K>(n+2)/4$, in which case
$K/\sqrt n\to\infty$, or else the displayed lower bound applies. If in addition
$\alpha_n=o(n)$, then
\begin{equation}
\sqrt{\alpha_n(4n+\alpha_n+4)}=2\sqrt{n\alpha_n}\,(1+o(1)),
\end{equation}
and since $\alpha_n=o(\sqrt{n\alpha_n})$, we obtain
\begin{equation}
K\ge \frac14\sqrt{n\alpha_n}\,(1+o(1)).
\end{equation}
This concludes the proof.  
\end{IEEEproof}

We are now ready to prove Theorem~\ref{thm:main-support}. 
\begin{IEEEproof}
The first two bounds in \eqref{eq:explicit-support} were shown in \cite{zieder2024binomial}. It remains to prove the third bound. Starting with \eqref{eq:upper_chi2}, we have that 
\begin{equation}
  \chi^2(P_{Y^\ast} \| P_{Y_r}) \le \zeta(c_\star^{-1}) \frac{17}{\log\!\left(\frac{n\pi}{2e}\right)}=: u_n,
\end{equation}
and, in addition, let 
\begin{equation}\label{eq:def_U_alphan}
 \alpha_n := \log \frac{1}{4u_n}.
\end{equation}
For $n > \frac{2}{\pi} e^{37851}$, we have $u_n < 1/4$ and $\alpha_n>0$. Therefore, by applying Proposition~\ref{prop:invert-chi2}, and by assuming $K \le (n+2)/4$, we have
\begin{align}
    K
&\ge
\frac{
4 - \alpha_n + \sqrt{\alpha_n(4n + \alpha_n + 4)}
}{8} \\
&\ge \frac{
2\sqrt{n\alpha_n} - \alpha_n
}{8} \label{eq:use_alphan_lessthan_n_1} \\
&\ge \frac{1}{8}\sqrt{n\alpha_n} \label{eq:use_alphan_lessthan_n_2} \\
&\ge\frac{1}{8}
\sqrt{
n \,
\log\!\left(
\frac{\log\!\left(\frac{n\pi}{2e}\right)}{68\,\zeta(c_\star^{-1}) }
\right)} \label{eq:use_def_alphan} \\
&\ge\frac{1}{8}
\sqrt{
n \,
\log\!\left(
\frac{\log\!\left(\frac{n\pi}{2e}\right)}{3268}
\right) 
}, \label{eq:lower_K}
\end{align}
where \eqref{eq:use_alphan_lessthan_n_1} and \eqref{eq:use_alphan_lessthan_n_2} follow from $\alpha_n \le n$ for $n> \frac{2}{\pi} e^{37851}$, which also implies $\alpha_n \le \sqrt{n\alpha_n}$; and \eqref{eq:use_def_alphan} follows from \eqref{eq:def_U_alphan}. 

Thus, for all $n >\frac{2}{\pi} e^{37851}$, either
$K > \frac{n+2}{4}$ or $K$ is lower-bounded by \eqref{eq:lower_K}. Since \eqref{eq:lower_K} is smaller than $(n+2)/4$ in this regime, the desired bound follows.  This concludes the proof.
\end{IEEEproof}
\vspace{-0.5cm}
\section{Conclusion}\label{sec:conclusion}
We have shown that the support size of the capacity-achieving input distribution of the binomial channel must grow faster than the previously known order $\sqrt n$. More precisely, we proved a lower bound of order $\sqrt{n\log\log n}$, up to explicit constants. The proof combines a quantitative lower bound on the $\chi^2$ divergence from a reference Beta-binomial output with new upper and lower bounds on the capacity with vanishing gap.

Several questions remain open. The exact asymptotic growth of the support size is still unknown, and there remains a substantial gap between the lower bound proved here and the best known upper bound of order $n/2$. It would also be of interest to refine the explicit constants, obtain sharper non-asymptotic estimates on the output distribution ratio $P_{Y^\ast}/P_{Y_r}$, and investigate whether the approximation viewpoint developed here can be transferred to related discrete-output channels.

\bibliographystyle{IEEEtran}
\bibliography{refs.bib}

\end{document}